\documentclass[copyright,creativecommons]{eptcs}
\providecommand{\event}{Automata \& JAC 2012}

\usepackage{amsthm,enumerate,amsmath,amscd,amssymb}
\usepackage{tikz}

\usepackage{calc}

\usepackage{comment}

\newtheorem{lemma}{Lemma}
\newtheorem{definition}{Definition}
\newtheorem{corollary}{Corollary}
\newtheorem{proposition}{Proposition}
\newtheorem{example}{Example}
\newtheorem{theorem}{Theorem}
\newtheorem{conjecture}{Conjecture}
\newtheorem{question}{Question}

\newcommand{\Z}{\mathbb{Z}}
\newcommand{\N}{\mathbb{N}}
\newcommand{\A}{\mathcal{A}}
\newcommand{\B}{\mathcal{B}}

\newcommand{\F}{\mathcal{F}}
\newcommand{\SIGMA}{\mathrm{\Sigma}}
\newcommand{\PI}{\mathrm{\Pi}}

\newcommand{\diamondlord}[4]{
\fill[#4,opacity=.15] (#1+#3,#2) -- (#1,#2+#3) -- (#1-#3,#2) -- (#1,#2-#3);
\draw[#4,thick] (#1+#3,#2) -- (#1,#2+#3) -- (#1-#3,#2) -- (#1,#2-#3) -- cycle;
}

\title{On Derivatives and Subpattern Orders of Countable Subshifts\thanks{Research supported by the Academy of Finland Grant 131558}}
\def\titlerunning{On Derivatives and Subpattern Orders of Countable Subshifts}

\author{
	Ville Salo
		\institute{TUCS -- Turku Centre for Computer Science, Finland, \\
		University of Turku, Finland}
		\email{vosalo@utu.fi}
	\and
	Ilkka T\"orm\"a
		\institute{University of Turku, Finland}
		\email{iatorm@utu.fi}
}
\def\authorrunning{Ville Salo, Ilkka T\"orm\"a}

\begin{document}
\maketitle

\begin{abstract}
We study the computational and structural aspects of countable two-dimensional SFTs and other subshifts. Our main focus is on the topological derivatives and subpattern posets of these objects, and our main results are constructions of two-dimensional countable subshifts with interesting properties. We present an SFT whose iterated derivatives are maximally complex from the computational point of view, a sofic shift whose subpattern poset contains an infinite descending chain, a family of SFTs whose finite subpattern posets contain arbitrary finite posets, and a natural example of an SFT with infinite Cantor-Bendixon rank.
\end{abstract}


\section{Introduction}

In this article, we study the computational and structural aspects of countable two-dimensional SFTs, with an emphasis on properties of the topological derivative and the so-called subpattern order. Our approach is mainly constructive, in that our main results are examples of subshifts with interesting properties. The study of computational aspects of tilings started with the observation of Wang that from a seed tile, one can simulate a Turing machine by simply drawing its run on the set of tilings. It was conjectured by Wang that without a seed tile, computation cannot be forced, and even that there in fact exists a periodic tiling. The tile sets of \cite{Be66} and \cite{Ro71} provide counterexamples, and further show that a Turing machine can be forced to run on every tiling by using a self-similar construction. In the case that there exist only countably many tilings, the situation is different: since every dynamical system contains a minimal subsystem (which is then the orbit closure of a single point), a countable subshift must contain a periodic point. However, Turing machines can still be run on such tilings, so that configurations not containing a seed tile form a small recursive set. From this observation, a wealth of interesting behavior emerges \cite{JeVa11}.

In \cite{BaDuJe08}, the notion of topological derivative is shown to be very useful for studying structural properties of countable two-dimensional SFTs. Namely, it is straightforward that such an SFT will eventually become empty when the derivative is iterated transfinitely, and at the very end of this process, simple structure must emerge: Unless the SFT in question is in fact finite, the second-to-last nonempty derivative must contain a point with exactly one vector of periodicity, and the last level must contain at least two fully periodic points.

The main open problem in \cite{BaDuJe08} about topological derivatives asks for a countable two-dimensional SFT with infinite rank. This has been completely solved in \cite{JeVa11}, and we give an independent weaker solution to this problem, hopefully showing a more natural example of how this type of behavior may occur in a countable two-dimensional SFT.

Another notion studied in \cite{BaDuJe08} is the partial order induced by subpattern inclusion: we say $x \succcurlyeq y$ if and only if all patterns seen in $y$ are also seen somewhere in $x$. Not much is known about this partial order in the class of countable two-dimensional SFTs, in particular (as far as we are aware) it is still open whether there exists such an SFT with an infinite descending chain. We solve this in the positive for two-dimensional sofic shifts, by finding such a chain in a countable sofic shift. While we cannot solve the descending chain problem for SFTs, we find rich structure in this class of partial orders by order-embedding every finite partially ordered set in the partial order of a countable SFT.

In addition to addressing the questions of \cite{BaDuJe08}, we study the computational complexity of the derivative of a countable two-dimensional SFT. The extension problem of such an SFT (that is, solving whether a given pattern $P$ occurs in a valid configuration) is easily seen to be $\PI^0_1$, and the complexity of this problem may increase by at most two levels in the arithmetical hierarchy when the derivative is taken. We show a converse to this: there exists a countable two-dimensional SFT whose $k$th derivative is $\PI^0_{2k+1}$-complete for all $k$. This implies that while the very last levels of the derivation process have simple structure, during the process the complexity may rise arbitrarily high and only decrease in a limit ordinal.

The structure of our paper is as follows. In Section~\ref{sec:Definitions}, we give the relevant definitions and notation used in the rest of the article. In Section~\ref{sec:Ranks}, we give our results about ranks attainable from subshifts in one and two dimensions. We solve the one-dimensional sofic case, which is drastically different from its two-dimensional counterpart. For the case of countable two-dimensional SFTs, we obtain transfinite rank with a natural subshift, which does not directly involve any kind of encoding of computation. In Section~\ref{sec:DerivativeComputability}, we show that derivatives of countable two-dimensional SFTs can climb arbitrarily high in the arithmetical hierarchy. In Section~\ref{sec:SubpatternOrder}, we order-embed an arbitrary finite poset in the subpattern poset of a countable two-dimensional SFT, and find a countable sofic shift with an infinite descending chain.

\section{Definitions and Notation}
\label{sec:Definitions}

Let $S$ be a finite set of \emph{symbols}, called the \emph{alphabet}, endowed with the discrete topology. For an integer dimension $d \geq 1$, the set $S^{\Z^d}$, equipped with the product topology, is called the \emph{$d$-dimensional full shift on $S$}. Elements $x$ of $S^{\Z^d}$ are called \emph{configurations}. A configuration $x \in S^{\Z^d}$ is \emph{unary} if there exists $s \in S$ with $x_{\vec n} = s$ for all $\vec n \in \Z^d$. A \emph{pattern over $S$} is a pair $(D,s)$, where $D \subset \Z^d$ is a finite \emph{domain}, and $s : D \to S$ gives the arrangement of symbols in $D$. A pattern $P = (D,s)$ \emph{occurs} in a configuration $x$, denoted $P \sqsubset x$, if we have $x_{D + \vec{n}} = P$ for some $\vec{n} \in \Z^d$. For all $k \in [1,d]$, we define the \emph{shift map} $\sigma_k : S^{\Z^d} \to S^{\Z^d}$ by $\sigma_k(x)_{\vec n} = x_{\vec n + e_k}$, where $\{e_1,\ldots,e_d\}$ is the natural generator set of $\Z^d$.

A \emph{$d$-dimensional subshift over $S$} is a closed subset $X \subset S^{\Z^d}$ satisfying $\sigma_k(X) = X$ for all $k \in [1,d]$. Alternatively, all subshifts $X$ can be defined by a set $\F$ of \emph{forbidden patterns} as $X = \{ x \in S^{\Z^d} \;|\; \forall P \in \F: P \not\sqsubset x \}$. If $\F$ is finite, then $X$ is said to be \emph{of finite type} (SFT for short). Given a finite domain $D \subset \Z^d$, the set of patterns occurring in the points of a subshift $X$ with domain $D$ is denoted $\B_D(X)$, the set of all patterns of $X$ is $\bigcup_D \B_D(X) = \B(X)$, and the set of symbols occurring in $X$ is denoted $\A(X)$. A \emph{block map} is a continuous mapping $f : X \to Y$, where $X$ and $Y$ are $d$-dimensional subshifts (possibly over different alphabets), which intertwines the shift maps of $X$ and $Y$: $f \circ \sigma_k = \sigma_k \circ f$ for all $k \in [1,d]$. Alternatively, a block map $f$ can be defined by a \emph{local function} $F : \B_D(X) \to \A(Y)$ by $f(x)_{\vec n} = F(x)_{D + \vec n}$ for all $x \in X$ and $\vec n \in \Z^d$, where $D$ is a finite domain, called the \emph{neighborhood} of $f$ \cite{He69}. An image of a subshift under a block map is a subshift, and images of SFT's are called \emph{sofic shifts}. See \cite[Section 13.10]{LiMa95} for a short survey on multidimensional symbolic dynamics.

A \emph{preordered set} is a tuple $(S,\geq)$, where $\geq$ is a binary relation on $S$ which is reflexive ($x \geq x$ holds for all $x$) and transitive ($x \geq y$ and $y \geq z$ imply $x \geq z$). A \emph{partial order} is a preorder which is antisymmetric ($x \geq y$ and $y \geq x$ imply $x = y$). A partially ordered set is called a \emph{poset}. An \emph{order-embedding} between two partially ordered sets $(S,\geq)$ and $(T,\succcurlyeq)$ is a function $f : S \to T$ such that for all $x,y \in S$ we have $x \geq y$ iff $f(x) \succcurlyeq f(y)$.

We define a preorder, called the \emph{subpattern order}, on the configurations of $S^{\Z^d}$ by stating that $x \succcurlyeq y$ holds iff $P \sqsubset y$ implies $P \sqsubset x$ for all patterns $P$, meaning that $x$ contains all patterns of $y$. This notion was first introduced in \cite{Du99}. If $x \succcurlyeq y$ and $x \preccurlyeq y$, we denote $x \approx y$, and if $x \succcurlyeq y$ and $x \not\approx y$, we denote $x \succ y$. The \emph{subpattern poset} of a subshift $X \subset S^{\Z^d}$ is the poset $(X/\!\!\approx,\succcurlyeq)$, where $\approx$-equivalent elements of $X$ are identified.

Given a topological space $X$, the \emph{Cantor-Bendixon derivative of $X$} (see e.g. \cite{Ku66}, first considered for SFTs in \cite{BaDuJe08}) is defined as $X' = \{ x \in X \;|\; x \in \overline{X - \{x\}} \}$. Thus $X'$ consists of the nonisolated points of $X$. Note that a nonisolated point of $X$ might become isolated in $X'$, and thus it makes sense to inductively define the $\lambda$th derivative of $X$ for all ordinals $\lambda$. First, $X^{(0)} = X$. If $\lambda = \alpha + 1$, then $X^{(\lambda)} = (X^{(\alpha)})'$, and if $\lambda$ is a limit ordinal, then $X^{(\lambda)} = \bigcap_{\alpha < \lambda} X^{(\alpha)}$. The \emph{rank} of $X$ is the least ordinal $\lambda$ such that $X^{(\lambda)} = X^{(\lambda + 1)}$. If $X$ is a subshift, then so is $X^{(\lambda)}$ for all $\lambda$, and if $X' \subsetneq X$, then $X'$ contains strictly less patterns than $X$. Since the set of all patterns is countable, the rank of $X$ exists and is a countable ordinal.

Let $\phi$ be a formula in first-order arithmetic. If $\phi$ contains only bounded quantifiers, then we say $\phi$ is $\SIGMA^0_0$ and $\PI^0_0$. For all $n > 0$, we say $\phi$ is $\SIGMA^0_n$ if it is equivalent to a formula of the form $\exists k: \psi$ where $\psi$ is $\PI^0_{n-1}$, and $\phi$ is $\PI^0_n$, if it is equivalent to a formula of the form $\forall k: \psi$ where $\psi$ is $\SIGMA^0_{n-1}$. This classification is called the \emph{arithmetical hierarchy} (see e.g. \cite[Chapter IV.1]{Od89} for an introduction to the topic). A subset $X$ of $\N$ is $\SIGMA^0_n$ or $\PI^0_n$, if $X = \{ x \in \N \;|\; \phi(x) \}$ for some $\phi$ with the corresponding classification. It is known that the $\SIGMA^0_1$ sets are exactly the recursively enumerable sets, and the $\PI^0_1$ sets their complements. When classifying sets of objects other than natural numbers (e.g. patterns), we assume that the objects are in some natural and computable bijection with $\N$. Also, a subshift is given the same classification as its language, so that, for example, two-dimensional SFTs are $\PI^0_1$ subshifts. See \cite{CeRe98} for a general survey on $\PI^0_1$ sets. The nonstandard quantifier $\exists^\infty n : \phi(n)$ has the meaning `there exist infinitely many $n$ such that $\phi(n)$.'

A subset $X \subset \N$ is \emph{many-one reducible} (or simply \emph{reducible}) to another set $Y \subset \N$, if there exists a computable function $f : \N \to \N$ such that $x \in X$ iff $f(x) \in Y$. If every set in a class $\mathcal{C}$ is reducible to $X$, then $X$ is said to be \emph{$\mathcal{C}$-hard}. If, in addition, $X$ is in $\mathcal{C}$, then $X$ is \emph{$\mathcal{C}$-complete}.

In the proof of one of our results, we utilize \emph{counter machines}, which we define here informally (\cite{Mi67} is the classical reference). A counter machine $M$ consists of a finite state set $\Sigma$ and a finite set of counters, each of which holds a value in $\N$. On a single step, the machine increments or decrements some of its counters by $1$ and goes to a new state $s \in \Sigma$, depending on its previous state and which of its counters contained the value $0$. This action may also be nondeterministic. A counter machine can be used to simulate a Turing machine, and thus to execute any algorithm.

\section{Ranks of Subshifts}
\label{sec:Ranks}

Ranks of subshifts have usually been studied in the countable case, where we have the following basic result:

\begin{lemma}[\cite{BaDuJe08}]
\label{lem:CountableRanked}
A subshift $X$ has $X^{(\lambda)} = \emptyset$ for some ordinal $\lambda$ if and only if it is countable.
\end{lemma}

We note here that if a configuration $x$ is isolated in a subshift $X$, then there exists a pattern $P = (D,s) \in \B(X)$ such that $x$ is the only element of $X$ with $x_D = P$. We use this intuition in many of our proofs.

First, we look at one-dimensional sofic shifts. They have a useful well-known characterization in terms of the different contexts of words that appear in them. We give the characterization without proof, but for example, it easily follows from Theorem 3.2.10 of \cite{LiMa95}.

\begin{definition}
The \emph{context} of a word $v \in S^*$ in a subshift $X \subset S^\Z$ is $C_X(v) = \{ (w,w') \in (S^*)^2 \;|\; wvw' \sqsubset X \}$.
\end{definition}

\begin{lemma}
A subshift $X \subset S^\Z$ is sofic if and only if it has a finite number of different contexts.
\end{lemma}

We now relate the contexts of words in a subshift and its derivative. While we prove this result only in the case of one-dimensional subshifts, since this is exactly what we need, the obvious generalization holds for subshifts of all dimensions.

\begin{lemma}
\label{lemma:SameContexts}
For a subshift $X \subset S^\Z$,
\[ C_X(u) = C_X(v) \implies C_{X^{(1)}}(u) = C_{X^{(1)}}(v). \]
\end{lemma}

\begin{proof}
Let $C_X(u) = C_X(v)$, and suppose that $(w, w') \in C_{X^{(1)}}(u) - C_{X^{(1)}}(v)$. Then $wvw' \not\sqsubset X^{(1)}$, so the set of points $x \in X$ with $x_{[0,|wvw'|-1]} = wvw'$ is finite. But since $C_X(u) = C_X(v)$, these are in a bijective correspondence with the points $y$ such that $y_{[0,|wuw'|-1]} = wuw'$, which implies that $wuw' \not\sqsubset X^{(1)}$, a contradiction. 
\end{proof}

This implies that the number of different contexts cannot increase in the derivative, so the two previous lemmas give the following:

\begin{corollary}
\label{cor:1DSoficNiceness}
The derivative of a one-dimensional sofic shift is sofic.
\end{corollary}

A simple further analysis proves the following result.

\begin{proposition}
\label{prop:SoficFinite}
All one-dimensional sofic shifts have finite rank.
\end{proposition}

\begin{proof}
Let $X$ be sofic, and let $k$ be the number of different contexts in $X$. If $X \neq X^{(1)}$, then necessarily $X \neq S^\Z$, so we may choose $u \notin \B(X)$. Further, choose $v \in \B(X) - \B(X^{(1)})$. Now, $C_X(u) \neq C_X(v)$, but $C_{X^{(1)}}(u) = \emptyset = C_{X^{(1)}}(v)$. By Lemma~\ref{lemma:SameContexts}, $X^{(1)}$ has at most $k - 1$ different contexts. By induction, $X^{(i)}$ has at most $k - i$ different contexts, and it is then clear that $X^{(i)} = X^{(i+1)}$ for some $i \leq k$. 
\end{proof}

It was asked in \cite{BaDuJe08} whether the rank of a countable two-dimensional SFT can be infinite. This problem was completely solved in \cite{JeVa11}, where the possible ranks were proven to be exactly those of $\PI^0_1$ subshifts. Thus the two-dimensional situation is in stark contrast with Proposition~\ref{prop:SoficFinite}. However, to our knowledge, there does not exist an example where an infinite rank arises `naturally', that is, from some simple geometric construction. We prove a weaker version of \cite[Theorem 4.4]{JeVa11} in Example~\ref{ex:Diamonds} with such a natural example, using no direct encoding of a computation.

\begin{example}
\label{ex:Diamonds}
There exists a countable two-dimensional SFT $X$ of rank at least $\omega$.
\end{example}

\begin{proof}
Consider the one-dimensional subshift containing points of the form
\[ {}^\infty 0 a^k 0^{m_1} a^{k-1}b 0^{m_2} a^{k-2}b^2 0^{m_3} \cdots 0^{m_k} b^k 0^\infty, \]
where $k \in \N$ and $m_i \in \N$ for all $i$ are arbitrary. For all $k$, the subshift contains configurations with $k$ `islands' floating in a sea of $0$'s, but no configuration contains an infinite number of islands. This is a countable subshift with infinite rank, and in the following, we construct a two-dimensional SFT $X$ that uses exactly the same idea.

The SFT $X$ contains one infinite horizontal \emph{dedicated line}. The top and bottom halves are colored differently. On the line one may have (perhaps infinite) \emph{diamonds}, colored red and blue, whose left and right corners must be on the dedicated line. The diamonds must be nested, that is, a blue diamond must either contain a red diamond or be contained in one (not both) and vice versa. This is established by sending signals along the dedicated line. Two distinct diamonds may not overlap, unless one is completely inside the other (including a complete overlap). The insides of the diamonds are colored differently from their outsides.

From the top (bottom) corner of every red (blue) diamond, a \emph{decrement signal} is sent to the right (left, respectively). Also, the top (bottom) corner of every red (blue) diamond must absorb one decrement signal traveling one tile above (below) it. The area between the line and a signal is colored differently from its complement. See Figure~\ref{fig:Diamonds} for a clarifying picture.

\begin{figure}[ht]

\begin{center}
\begin{tikzpicture}


\diamondlord{2.4}{0}{1.55}{blue};
\diamondlord{2.5}{0}{1.15}{red};
\diamondlord{6.7}{0}{1.8}{blue};
\diamondlord{6.9}{0}{.9}{red};

\draw[blue,thick,dashed,->] (9,-2.05) -- (6.7,-2.05) -- (6.7,-1.8);
\draw[blue,thick,dashed,->] (6.7,-1.8) -- (2.4,-1.8) -- (2.4,-1.55);
\draw[blue,thick,dashed] (2.4,-1.55) -- (0,-1.55);

\draw[red,thick,dashed,->] (0,1.4) -- (2.5,1.4) -- (2.5,1.15);
\draw[red,thick,dashed,->] (2.5,1.15) -- (6.9,1.15) -- (6.9,.9);
\draw[red,thick,dashed] (6.9,.9) -- (9,.9);

\draw[thick] (0,0) -- (9,0);

\end{tikzpicture}
\end{center}

\caption{The diamond construction.}
\label{fig:Diamonds}
\end{figure}
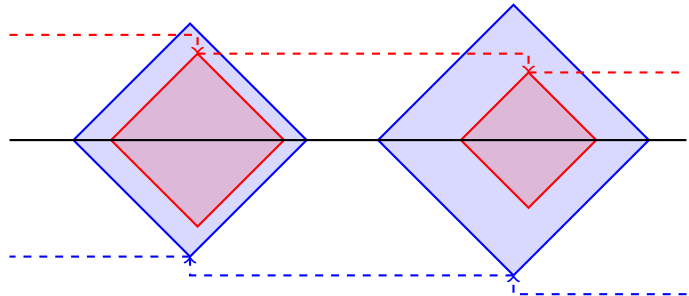

We first show that $X$ is countable. Indeed, for each $(n,m) \in \N^2$, if a configuration $x$ of $X$ contains nested diamonds of sizes $n$ and $m$, then there are at most $n+m-1$ dedicated points in $x$, since the size of the red (blue) diamonds decreases to the right (left). The number of ways to arrange these points and the surrounding diamonds is countable. One can also check that the number of exceptional points (ones containing, say, an infinite diamond or just signals) is countable.

Next, we show that $X^{(\omega)}$ is a nonempty set of finite rank. First, the isolated points of $X$ are exactly those that contain finite red and blue diamonds, and whose rightmost red and leftmost blue diamonds are of size $1$. In general, if $x \in X$ contains finite red and blue diamonds, we say that $x$ \emph{has type $(n,m)$} if the rightmost blue (leftmost left) is of size $n$ ($m$, respectively). It is then easy to see that for all $k \in \N$, the set $X^{(k)}$ will contain all of $X$, except for the points of type $(n,m)$ with $n + m < k$. Then $X^{(\omega)}$ is nonempty, but will consist of only the exceptional points, and clearly $X^{(\omega+k)} = \emptyset$ for some finite $k$. 
\end{proof}

Since the rank of a Cartesian product of subshifts is the Hessenberg sum of their ranks \cite{Ku66} and the product of SFTs is again an SFT, we have, for all $k \in \N$, a countable SFT with rank at least $\omega \cdot k$. This is of course still much weaker than the result of \cite{JeVa11}, where arbitrarily large recursive ordinals were obtained as ranks of two-dimensional SFTs.

\section{Computability Aspects of Derivatives}
\label{sec:DerivativeComputability}

In this section, we study the computational power of the $k$th derivative of a countable two-dimensional SFT, which turns out to possibly climb very high in the arithmetical hierarchy. We start with an upper bound, which we then reach with a construction. A generalization of the following lemma was proved in \cite[Lemma 1.2 (3)]{CeClSmSoWa86}, but we include a proof for completeness.

\begin{lemma}
\label{lem:IncByTwo}
Given a two-dimensional $\PI^0_k$ subshift $X$ and a pattern $P$, it is $\PI^0_{k+2}$ whether $P \sqsubset X^{(1)}$.
\end{lemma}

\begin{proof}
Given $X$ and $P$, we have $P \sqsubset X^{(1)}$ iff $P \sqsubset X$ and for all $n \in \N$, there exist two distinct equal-sized extensions $Q_1, Q_2 \sqsubset X$ of $P$ that agree on the $(n \times n)$-square around $P$. This is clearly $\PI^0_{k+2}$. 
\end{proof}

The following construction shows that the bound given by Lemma~\ref{lem:IncByTwo} on the complexity of $k$th derivatives of $\PI^0_1$ subshifts is strict, and can be attained by a single countable SFT. In particular, it implies that Corollary~\ref{cor:1DSoficNiceness} fails miserably in higher dimensions, since two-dimensional sofic shifts are $\PI^0_1$, while their derivatives may be $\PI^0_3$-complete, and thus highly nonsofic. The rank of the subshift we build will be $\omega + k$ for some finite $k$, and with slight modifications, we could guarantee its $\omega$th derivative to be recursive (with this exact construction, it is probably already recursive, but we have not verified this). We start with a definition, and a classical computability lemma.

\begin{definition}
For $k \in \N$, denote by $\Phi_k$ the set of first-order arithmetical formulas with $k$ free variables and only bounded quantifiers. For $k,l \in \N$, denote by $\phi^k_l$ the $l$th formula in $\Phi_k$, ordered first by length and then lexicographically.
\end{definition}

\begin{lemma}[Lemma 2 in \cite{KrShWa60}]
\label{lemma:Infty}
Let $k \in \N$ and $\phi \in \Phi_{2k+1}$. Then there exists $\psi \in \Phi_{k+1}$, uniformly computable from $\phi$ and $k$, such that
\[ \forall n_1 : \exists n_2 : \cdots \forall n_{2k-1} : \exists n_{2k} : \forall n_{2k+1} : \phi(n_1, \ldots, n_{2k+1}) \]
is equivalent to
\[ \exists^\infty n_1 : \exists^\infty n_2 : \cdots \exists^\infty n_k : \forall n_{k+1} : \psi(n_1, \ldots, n_{k+1}). \]
\end{lemma}

We denote $\psi = I(\phi)$ in the above lemma. With this result, we can transform alternating quantifiers into infinitary ones, and the application to derivatives is rather straightforward.

\begin{theorem}
\label{thm:DerivativesCanCompute}
There exists a countable two-dimensional SFT $X$ for which the problem whether $P \sqsubset X^{(k)}$ for a given pattern $P$ is $\PI^0_{2k+1}$-complete, for all $k \in \N$.
\end{theorem}

\begin{proof}
Consider the closure of the subset of $\{0,1,2\}^\N$ consisting of points of the form
\[ 1^l 2^k 0^{n_1} 1 0^{n_2} \cdots 0^{n_k} 1 0^\infty, \]
where $I(\phi_l^{2k+1})(n_1, n_2, \ldots, n_k, n_{k+1})$ is true for all $n_{k+1}$. This set is $\PI^0_1$-complete. Clearly, the derivative of this closed set contains only those points of the form
\[ 1^l 2^k 0^{n_1} 1 0^{n_2} \cdots 0^{n_{k-1}} 1 0^\infty, \]
where $I(\phi_l^{2k+1})(n_1, n_2, \ldots, n_{k-1}, n_k, n_{k+1})$ holds for infinitely many $n_k$ and all $n_{k+1}$. Thus the derivative is $\PI^0_3$-complete, and we could verify by induction that the $n$th derivative is $\PI^0_{2n+1}$-complete. The construction of $X$ is an implementation of the same idea by a two-dimensional SFT.

A typical configuration $x \in X$ consists of the \emph{input}, a segment of the form $1^l2^k$ extending to the right from the origin, and the \emph{computation area}, a filled cone extending upwards from the input. The input also extends upwards in order to be accessible in the computation area. The rest of $x$ is filled with $0$'s. See Figure~\ref{fig:Typicalx} for a visualization.

\begin{figure}[ht]

\begin{center}
\begin{tikzpicture}[scale=0.15]           

\draw[->] (-3,1.5) -- (-1,1.5);
\draw[->] (-3,9.5) -- (-1,9.5);
\draw[->] (-3,25.5) -- (-1,25.5);

\foreach \y in {1,...,30}{
	\filldraw[fill=white] (0,\y) rectangle +(1,1);
}
\foreach \y in {2,...,30}{
	\filldraw[fill=red!70!white] (1,\y) rectangle +(1,1);
	\foreach \x in {2,...,6}{
		\filldraw[fill=black!30] (\x,\y) rectangle +(1,1);
	}
	\foreach \x in {7,...,9}{
		\filldraw[fill=black!50] (\x,\y) rectangle +(1,1);
	}
	\foreach \x in {10,...,39}{
		\filldraw[fill=white] (\x,\y) rectangle +(1,1);
	}
	\foreach \x in {2,...,\y}{
		\filldraw[fill=green!50!white] (\x+8,\y) rectangle +(1,1);
	}
	\filldraw[fill=red!70!white] (\y+8,\y) rectangle +(1,1);
}
\foreach \x in {2,...,6}{
	\filldraw[fill=black!30] (\x,1) rectangle +(1,1);
}
\foreach \x in {7,...,9}{
	\filldraw[fill=black!50] (\x,1) rectangle +(1,1);
}
\foreach \x in {10,...,39}{
	\filldraw[fill=white] (\x,1) rectangle +(1,1);
}

\foreach \x in {0,...,39}{
	\filldraw[fill=white] (\x,0) rectangle +(1,1);
}

\foreach \y in {1,...,6}{
	\filldraw[fill=blue,opacity=0.6] (\y*2+1,\y+1) rectangle +(1,1);
}
\foreach \x in {1,...,14}{
	\filldraw[fill=blue,opacity=0.6] (\x+1,8) rectangle +(1,1);
}

\foreach \y in {1,...,14}{
	\filldraw[fill=blue,opacity=0.6] (\y*2+1,\y+9) rectangle +(1,1);
}
\foreach \x in {1,...,30}{
	\filldraw[fill=blue,opacity=0.6] (\x+1,24) rectangle +(1,1);
}

\foreach \y in {1,...,5}{
	\filldraw[fill=blue,opacity=0.6] (\y*2+1,\y+25) rectangle +(1,1);
}

\end{tikzpicture}
\end{center}

\caption{A typical point of $X$. The blue lines represent the bouncing ball. Counters and states, whose values are updated at the lines indicated by an arrow, are not shown here.}
\label{fig:Typicalx}
\end{figure}
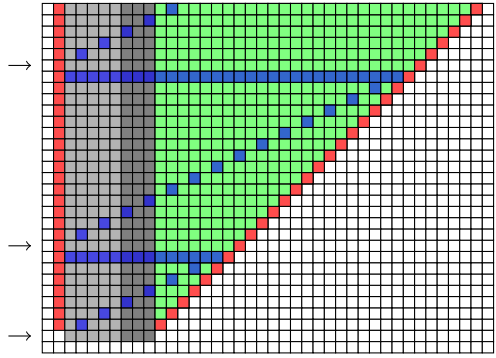

Inside the computation area, a \emph{ball} bounces between the walls of the cone, and in every sweep, one step of a counter machine $M$ is simulated. This is simply to ensure that $X$ is countable, as every nontrivial computation now has a starting point. The machine operates inside the cone, where the values of its counters are stored as the distances of special \emph{counter symbols} from the vertical line going through the origin. A counter $c$ with value $n$ is represented by a length-$n$ sequence of a symbol $\langle c \rangle$ extending right from the central column. Configurations of $X$ thus correspond to computation histories of $M$ in a concrete way.

Given the input $(l,k)$, the machine $M$ sequentially guesses $k$ natural numbers $n_1, \ldots, n_k$ and then checks in an infinite loop that $I(\phi^{2k+1}_l)(n_1,\ldots,n_k,n_{k+1})$ holds for all $n_{k+1} \in \N$. If the check fails at some point, a tiling error is produced. The guesses are also made using loops, so that for every $i \in [1,k]$, $M$ starts enumerating all $n \in \N$, and at some point decides that $n_i$ gets the value $n$. Thus larger guesses for the numbers take more time to compute, and since the computation is visible in the subshift, the input pattern $01^l2^k0$ occurs in $X^{(k)}$ if and only if
\[ \exists^\infty n_1 : \cdots \exists^\infty n_k : \forall n_{k+1} : I(\phi^{2k+1}_l)(n_1, \ldots, n_{k+1}) \]
is true. But by Lemma~\ref{lemma:Infty}, this is equivalent to
\[ \forall n_1 : \exists n_2 : \cdots \forall n_{2k-1} : \exists n_{2k} : \forall n_{2k+1} : \phi^{2k+1}_l(n_1, \ldots, n_{2k+1}), \]
and thus the subshift $X^{(k)}$ is $\PI^0_{2k+1}$-hard in the sense of the claim. Since it reaches the upper bound given by Lemma~\ref{lem:IncByTwo}, it is actually $\PI^0_{2k+1}$-complete.

The only thing left to prove is the countability of $X$. For each input pattern $01^l2^k0$, there are only a countable number of ways to complete it into a configuration, since the computation structure is forced, and only $k$ nondeterministic moves are made by $M$. Configurations which do not contain input, or which contain an infinite input, are degenerate, and a simple case analysis shows that they, too, form a countable set. 
\end{proof}

In \cite{JeVa11}, a similar encoding of `computation in a cone' is used, but instead of counter machines, the authors embed computation histories of Turing machines into configurations of countable SFTs. While both approaches have their merits, we feel that it is slightly more obvious how the counter machine construction works and why the resulting subshift is countable.

\section{Subpattern Order}
\label{sec:SubpatternOrder}

We now focus on the subpattern posets of countable multidimensional subshifts.

\begin{proposition}
\label{prop:NoUp}
A countable $d$-dimensional subshift does not contain an infinite upward chain with relation to $\succ$.
\end{proposition}

\begin{proof}
Suppose that a subshift $X \subset S^{\Z^d}$ contains a chain $(x^i)_{i \in \N}$ with $x^i \succ x^{i-1}$ properly for all $i \geq 1$. We show that $X$ is uncountable by constructing an injective map $f : \{0,1\}^\N \to X$. First, note that if some $x^i$ with $i \geq 1$ were periodic, then the relation $x^i \succ x^{i-1}$ could not be proper, and so all $x^i$ with $i \geq 1$ are aperiodic.

Define $k_0 = 0$, and consider the one-cell pattern $P(\epsilon) = x^0_0$, which must occur somewhere in $x^1$. Since $x^0 \prec x^1$ properly, $x^1$ is not in the orbit of $x^0$, and thus there exist two distinct patterns $P(0)$ and $P(1)$ of size $(2k_1+1)^d$ for some $k_1$ occurring in $x^1$ with $P(\epsilon)$ in the center. In general, for all $n \in \N$, there exists $k_n \in \N$ such that for all words $w \in \{0,1\}^{n-1}$, we have two distinct patterns $P(w0)$ and $P(w1)$ of size $(2k_n+1)^d$ occurring in $x^n$ and containing $P(w)$ in their centers.

For all $w \in \{0,1\}^\N$, we define $f(w) = \lim_{n \longrightarrow \infty} P(w_{[0,n-1]})$. Then $f$ is a well-defined injection from $\{0,1\}^\N$ to $X$, and the claim is proved. 
\end{proof}

This is a generalization of \cite[Theorem 3.7]{BaDuJe08}, which states the result for two-dimensional SFTs. While also their method directly generalizes to all countable subshifts, ours is more combinatorial in nature, and gives an explicit (although not necessarily effective) continuous injection from $\{0,1\}^\N$ to $X$. For antichains we have the following example.

\begin{example}
\label{ex:YesSide}
There exists a countable two-dimensional SFT with an infinite number of periodic points and an infinite antichain in its subpattern poset: take the SFT where horizontal and vertical lines form an infinite grid, and every rectangle is forced to be a square using a diagonal signal.
\end{example}

While no countable subshift contains an infinite ascending chain, and a simple countable SFT with an infinite antichain exists, the problem of descending chains turns out to be much more involved. We repeat the following, yet unsolved, conjecture from \cite{BaDuJe08}.

\begin{conjecture}
\label{conj:NoDown}
There is no countable two-dimensional SFT with an infinite downward chain for $\succ$.
\end{conjecture}

We will not prove this conjecture, but provide a counterexample in the sofic case. For this result, we use a lemma for simplicity's sake, although the subshift we construct with it also has a direct implementation using signals.

\begin{lemma}[\cite{DuRoSh10}]
\label{lemma:CoREAreSofic}
Let $Y \subset S^\Z$ be a one-dimensional $\PI^0_1$ subshift. Then $X = \{ x \in S^{\Z^2} \;|\; \exists y \in Y: \forall i,j \in \Z: x_{(i,j)} = y_j \}$ is a two-dimensional sofic shift.
\end{lemma}

\begin{theorem}
\label{thm:SoficChain}
There exists a countable two-dimensional sofic shift with an infinite decreasing chain with relation to $\succ$.
\end{theorem}

\begin{proof}
We construct the orbit closure of the following kind of binary configuration. On the positive $x$-axis, the configuration has a $1$ at each coordinate $2^n$. On each horizontal line at the heights $2^n$, between the coordinates where the previous line contains a $1$, a prefix of the pattern in the lowest line appears. The decreasing chain is obtained as follows: First, take the limit of the sequence where the configuration is shifted $2^n$ steps to the left. Now, the lowest line is reduced to a single $1$ in the origin, and the other lines contain a sparser `copy' of the original configuration. The chain is obtained by repeating this procedure on the successive lines.

We will explicitly construct said sofic shift by defining several layers, each of which is a sofic shift, placing them over each other with some constraints and finally applying a block map that forgets almost all of the data. We start with the one-dimensional binary \emph{powers of $2$ shift}, in which the patterns $10^{n-1}10^{m-1}a$ and $a0^{n-1}10^{m-1}1$ are only allowed if $m = 2n$ and $a = 1$. The subshift is clearly $\PI^0_1$, and is the orbit closure of the configuration with $1$'s in the coordinates $2^n$ for all $n > 0$. The corresponding two-dimensional shift, which contains infinite horizontal lines of $1$'s in a sea of $0$'s, is thus sofic by Lemma~\ref{lemma:CoREAreSofic}, and we use it as the basis of our construction. On these horizontal lines, we may put an arbitrary number of \emph{dedicated points}.

We then present the \emph{powers of $2$ gadget}, which will control the way in which the dedicated points appear. For the time being, concentrate for simplicity's sake on the bottom line. If two points $A,B$ lie on the line with distance $d$, we want a point $C$ to appear $2d$ steps to the right of $B$, and if $d > 1$, another point $D$ $\frac{d}{2}$ steps to the left of $A$. This is achieved with six signals emitted by each point, presented in Figure~\ref{fig:Signals}. Each point emits slope-$1$ and slope-$\frac{1}{2}$ signals to the left, a signal up and down, and slope-$1$ and slope-$2$ signals to the right. The downward signal is destroyed when the other signals hit it, and they must be correctly matched on each side. If $d = 1$ in the above situation, then $A$ only emits the rightward signals, and creates a \emph{forbidden zone} on its left. The zone continues infinitely upward and to the left, and no dedicated points may be situated inside it. Two signals of different types (dashed and dotted lines in the figure) may always cross each other and the horizontal lines, and a dashed line may cross a downward signal that has already encountered another dashed line. The gray area in the figure is another forbidden zone. It consists of the rightmost $\frac{2}{3}$ of the space between two points and continues infinitely upwards. Also, no forbidden zone may appear below the lowest line.

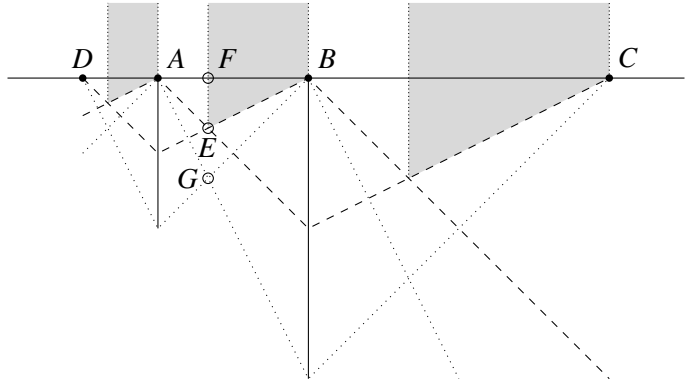
\begin{figure}[ht]

\begin{center}
\begin{tikzpicture}

\fill[white!85!black] (1+1/3,7) -- (2,7) -- (2,6) -- (1+1/3,5+2/3);
\fill[white!85!black] (2+2/3,7) -- (4,7) -- (4,6) -- (2+2/3,5+1/3);
\fill[white!85!black] (5+1/3,7) -- (8,7) -- (8,6) -- (5+1/3,4+2/3);

\draw (0,6) -- (9,6);
\fill (1,6) circle (0.05);
\fill (2,6) circle (0.05);
\fill (4,6) circle (0.05);
\fill (8,6) circle (0.05);

\draw (2+2/3,5+1/3) circle (0.07);
\draw (2+2/3,6) circle (0.07);
\draw (2+2/3,4+2/3) circle (0.07);

\draw (2,6) -- (2,4);
\draw[dashed] (1,6) -- (2,5) -- (4,6) -- (8,2);
\draw[dotted] (1,6) -- (2,4) -- (4,6) -- (6,2);

\draw (4,6) -- (4,2);
\draw[dashed] (1,5.5) -- (2,6) -- (4,4) -- (8,6);
\draw[dotted] (1,5) -- (2,6) -- (4,2) -- (8,6);

\node[above] () at (1,6) {$D$};
\node[above right] () at (2,6) {$A$};
\node[above right] () at (4,6) {$B$};
\node[above right] () at (8,6) {$C$};
\node[below] () at (2+2/3,5+1/3) {$E$};
\node[above right] () at (2+2/3,6) {$F$};
\node[left] () at (2+2/3,4+2/3) {$G$};

\draw[densely dotted] (1+1/3,5+2/3) -- (1+1/3,7);
\draw[densely dotted] (2,6) -- (2,7);
\draw[densely dotted] (2+2/3,5+1/3) -- (2+2/3,7);
\draw[densely dotted] (4,6) -- (4,7);
\draw[densely dotted] (5+1/3,4+2/3) -- (5+1/3,7);
\draw[densely dotted] (8,6) -- (8,7);

\end{tikzpicture}
\end{center}

\caption{The signals of the powers of $2$ gadget.}
\label{fig:Signals}
\end{figure}

We now drop the single line hypothesis, and require that all signals attempting to enter a forbidden zone are destroyed. On lines other than the first one, when an upward signal emitted by a point reaches the first horizontal line, it checks with a slope-$1$ \emph{check signal} (line segment $AB$ in Figure~\ref{fig:Forcing}) whether the (perhaps infinite) rectangle formed by the two lines and the left border of the next forbidden zone is at least as wide as it is high. If this is the case, it forces two consecutive points to appear one step to the right, and if not, the whole interval gets forbidden infinitely upwards. Figure~\ref{fig:Forcing} shows an interval in which the check succeeds, and the two points are forced to appear in the upper line.

\begin{figure}[ht]

\begin{center}
\begin{tikzpicture}[scale=1.5]

\fill[white!85!black] (1.2+.2/3,2-.2/3) -- (1.2+.2/3,3) -- (1.4,3) -- (1.4,2);
\draw[dotted] (1.2+.2/3,2-.2/3) -- (1.2+.2/3,3);
\draw[dotted] (1.4,3) -- (1.4,2);
\fill[white!85!black] (1.4+.4/3,2-.4/3) -- (1.4+.4/3,3) -- (1.8,3) -- (1.8,2);
\draw[dotted] (1.4+.4/3,2-.4/3) -- (1.4+.4/3,3);
\draw[dotted] (1.8,3) -- (1.8,2);
\fill[white!85!black] (1.8+.8/3,2-.8/3) -- (1.8+.8/3,3) -- (2.6,3) -- (2.6,2);
\draw[dotted] (1.8+.8/3,2-.8/3) -- (1.8+.8/3,3);
\draw[dotted] (2.6,3) -- (2.6,2);
\fill[white!85!black] (2.6+1.6/3,2-1.6/3) -- (2.6+1.6/3,3) -- (4.2,3) -- (4.2,2);
\draw[dotted] (2.6+1.6/3,2-1.6/3) -- (2.6+1.6/3,3);
\draw[dotted] (4.2,3) -- (4.2,2);

\fill[white!85!black] (0,3) -- (1.1,3) -- (1.1,2) --  (1,2) -- (1,1) -- (0,1/3);
\fill[white!85!black] (6,3) -- (5,3) -- (5,0) -- (6,0);
\draw[dotted] (1.1,3) -- (1.1,2);
\draw[dotted] (1,2) -- (1,1);
\draw[dashed] (1,1) -- (0,1/3);
\draw[dotted] (5,3) -- (5,0);

\draw (0,1) -- (6,1);
\draw (0,2) -- (6,2);
\fill (1,1) circle (0.05);
\fill (1.1,2) circle (0.05);
\fill (1.2,2) circle (0.05);
\fill (1.4,2) circle (0.05);
\fill (1.8,2) circle (0.05);
\fill (2.6,2) circle (0.05);
\fill (4.2,2) circle (0.05);

\draw[dotted] (1,2) -- (2,1);
\draw (2,1) circle (0.05);

\draw[dash pattern=on 1pt off 1pt] 
	(1.2,2) -- (1.2+.2/3,2-.2/3) -- 
	(1.4,2) -- (1.4+.4/3,2-.4/3) -- 
	(1.8,2) -- (1.8+.8/3,2-.8/3) -- 
	(2.6,2) -- (2.6+1.6/3,2-1.6/3) -- 
	(4.2,2) -- (5,2-.8);

\node[above left] () at (1,2) {$A$};
\node[above right] () at (2,1) {$B$};

\end{tikzpicture}
\end{center}

\caption{A successful check, which forces two points to appear on the upper line. The powers of $2$ gadget, in turn, forces more points to appear, until the interval ends. Not all signals are shown.}
\label{fig:Forcing}
\end{figure}
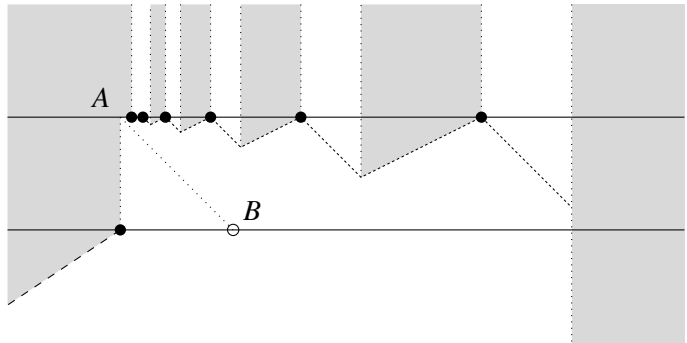

We then claim that signals emitted by points from different lines do not destructively interact with each other. Namely, if the distance of $A$ and $B$ in Figure~\ref{fig:Signals} is $d$, then the maximum distance between two points on the line $AF$ is at most $\frac{1}{2}\frac{1}{3}d = \frac{d}{6}$. Then the length of the downward signals emitted by these points is at most $\frac{d}{3}$, which is the length of the segment $EF$. Thus no signals emitted by these points propagate below the point $E$. Also, all emitted leftward signals are caught by downward ones. Finally, only the dotted signals and the downward signals that have already encountered a dashed signal are able to reach the dashed line $AE$, but these will just intersect without interaction, and no signal can reach the dotted line $AG$.

Consider then the sofic shift $X \subset \{0,1\}^{\Z^2}$ obtained by mapping every dedicated point to $1$, and the rest to $0$. We first prove that it is countable. Suppose first that a configuration of $X$ contains points in only one horizontal line. There might be forbidden zones extending infinitely downward from the line. On intervals between two such forbidden zones, no points are ever seen, since the check signals always fail. On an infinite interval free of forbidden zones, we either have an infinite sequence of points determined by the position of the leftmost one (which always exists), or just one lone point.

If points are seen on two lines, we know exactly the positions of all lines. Consider the leftmost point on the lowest occupied line. Because of the forbidden zone it creates on its left, it is the overall leftmost point. If there is another point on the lowest line, then all points on the line are determined by the gadget. Next, all points of the line above are determined by the upward signals and the gadget. By induction, all the points of the configuration are determined. If there are no other points on the lowest line, the next line is determined by the gadget, and we repeat the above argument. So all in all, a configuration is determined by the position of the lowest and leftmost point, and whether there are other points on the same line or some other line, and $X$ is countable.

Finally, we show that $X$ contains an infinite downward chain with relation to $\succ$. Consider the point $x_1$ containing infinitely many points on the lowest line. The gadgets force $x_1$ to consist of vertical `stripes' of exponential width that begin with a $1$ on the lowest line, followed by a prefix of the lowest line on the second one (and the patterns on the higher lines determined by it), and a forbidden zone. See Figure~\ref{fig:Forcing} to better visualize this. The prefixes become arbitrarily long as the width of the stripes increases. In the orbit closure of $x_1$, we thus find the point $x_2$ containing a $1$ in the lowest line followed by the stripe pattern on the first line of $x_1$ lifted to the second one, which creates a similar prefix pattern on the third line. We repeat the argument to find $x_3$, on which the striped pattern lies on the third line, and inductively we obtain the chain $(x_i)_{i \geq 1}$. Clearly $x_i \succ x_{i+1}$ holds properly for all $i$, since $x_{i+1}$ was chosen from the orbit closure of $x_i$, and the claim is proved. 
\end{proof}

We now consider the structure of finite subpattern posets. We first show that they do not capture the whole class of finite posets: for instance, no nontrivial lattice can occur as a subpattern poset. This is an easy consequence of the results in \cite{BaDuJe08}.

\begin{proposition}
A countable two-dimensional SFT whose subpattern poset is nontrivial contains two periodic points in distinct $\approx$-classes.
\end{proposition}

\begin{proof}
Assume the contrary, and let $X$ be a countable two-dimensional SFT which has exactly one periodic point $x \in X$ modulo $\approx$-equivalence, and let $r > 0$ be such that $x = \sigma_{(0,r)}(x) = \sigma_{(r,0)}(x)$. Let $m \in \N$ be the window size of $X$. If $X$ had a nontrivial subpattern poset, \cite{BaDuJe08} would imply the existence of a point $y \in X$ with exactly one direction of periodicity, say $(p,q) \in \Z^2$. We have two possibilities.
\begin{enumerate}
\item There exists $n \in \N$ such that every $(n \times n)$-square of $y$ contains a coordinate $d \in \Z^2$ such that the $(r \times r)$-square pattern of $y$ whose lower left corner is at $d$ does not appear in $x$. But then $x$ is not the only minimal point of $X$, and since $X$ contains only periodic minimal points, this is a contradiction.
\item For all $n \in \N$, an $(n \times n)$-square occurring in $x$ also occurs in $y$. In particular, this holds for $n = \max(2r|p|,2r|q|,2m)$. Since we have $\sigma_{(r|p|,r|q|)}(y) = y$, this implies that $y$ contains infinite strips of thickness at least $\max(r,m)$ consisting of the periodic pattern of $x$, and since $y \neq x$, the periodic pattern is broken between some of them. But now we can build an uncountable number of points in $X$ by joining an infinite number of these stripes together from the periodic areas, a contradiction.
\end{enumerate}
Thus, $X$ has a trivial subpattern poset.
\end{proof}

This means that if the subpattern poset $P$ of a countable two-dimensional SFT is nontrivial, then $P$ contains at least two minimal elements, and that if a countable two-dimensional SFT $X$ is infinite, it contains at least two periodic points in different orbits.

In \cite{BaDuJe08}, a proof sketch was given for the fact that for all $k \in \N$, the linearly ordered poset $(\{1, \ldots, k\}, <)$ can be order-embedded in the subpattern poset of a countable two-dimensional SFT. We generalize this with the following embedding result, for which we do not know an essentially simpler proof.

\begin{proposition}
\label{prop:OrderEmbed}
All finite posets can be order-embedded in the subpattern poset of some countable two-dimensional SFT. Furthermore, the subpattern poset itself can be made finite.
\end{proposition}

\begin{proof}
The idea of the construction is the following. Some configurations of the SFT correspond to elements of the poset. Such a configuration $x$ is either periodic, if the corresponding element lies at the bottom of the poset, and otherwise consists of infinitely many `boxes'. Each box contains a pattern from some other configuration whose poset element is lower than the element corresponding to $x$. The boxes and their contents are carefully aligned to ensure the countability of the SFT.

We now present the construction in more detail. Let $(S,\geq)$ be a finite poset, and for all $x \in S$, define
\[ r(x) = \max \{ n \in \N \;|\; \exists y_1, \ldots, y_n \in S: x = y_1 > \ldots > y_n \} - 1. \]
This is one less than the maximal length of a descending chain beginning from $x$. Define $M = r^{-1}(0)$, the set of minimal elements of $S$. Define also
\[ p(x) = \{ y \in S \;|\; x > y, \nexists z \in S: x > z > y \}, \]
the set of immediate predecessors of $x$. We also inductively define $k(x) = 1$ for all $x \in M$, and $k(x) = 1 + \sum_{y \in p(x)} k(y)$ for $x \notin M$. This is an auxiliary `height' function we need in our construction.

We build a two-dimensional SFT $X$ in whose subpattern poset $S$ can be order-embedded via $f : S \to X$. First, for each $x \in M$, $X$ contains a unary point $f(x)$. Let then $x \in S - M$. We assume that $f(y)$ has already been defined for all $y \in S$ with $r(y) < r(x)$, using the construction we are about to present if $y \notin M$.

The point $f(x)$ contains a horizontal \emph{dedicated half-line}, starting from the origin and extending right. Below the line there is a vertical sequence of \emph{ruler rectangles}, starting with one of size $1 \times 1$ below the origin. The $n$th rectangle, in general, has size $\phi(n,r(x)) \times n$, where $\phi(n,1) = n$ and $\phi(n,r) = \sum_{i=1}^n \phi(i,r-1)$. If $r(x) = 1$, this is achieved with a diagonal signal forcing the rectangles to be squares, and in general by stacking a sequence of lower-rank rectangles inside the large ones. See Figure~\ref{fig:Rectangles} for a visualization, and note that $\phi(n,r)$ is a polynomial of degree $r$ in $n$.

\begin{figure}[ht]

\begin{center}
\begin{tikzpicture}[rotate=-90]

\draw[thick,red] (.1,.1) -- (.1,9.9) -- (2.9,9.9) -- (2.9,.1) -- cycle;

\draw[thick,blue] (.2,.2) -- (.8,.2) -- (.8,.8) -- (.2,.8) -- cycle;
\draw[thick,blue] (.2,1.2) -- (1.8,1.2) -- (1.8,3.8) -- (.2,3.8) -- cycle;
\draw[thick,blue] (.2,4.2) -- (2.8,4.2) -- (2.8,9.8) -- (.2,9.8) -- cycle;

\draw[thick,black] (.3,.3) -- (.7,.3) -- (.7,.7) -- (.3,.7) -- cycle;
\draw[thick,black] (.3,.3) -- (.7,.7);
\draw[thick,black] (.3,1.3) -- (.7,1.3) -- (.7,1.7) -- (.3,1.7) -- cycle;
\draw[thick,black] (.3,1.3) -- (.7,1.7);
\draw[thick,black] (.3,4.3) -- (.7,4.3) -- (.7,4.7) -- (.3,4.7) -- cycle;
\draw[thick,black] (.3,4.3) -- (.7,4.7);
\draw[thick,black] (.3,2.3) -- (1.7,2.3) -- (1.7,3.7) -- (.3,3.7) -- cycle;
\draw[thick,black] (.3,2.3) -- (1.7,3.7);
\draw[thick,black] (.3,5.3) -- (1.7,5.3) -- (1.7,6.7) -- (.3,6.7) -- cycle;
\draw[thick,black] (.3,5.3) -- (1.7,6.7);
\draw[thick,black] (.3,7.3) -- (2.7,7.3) -- (2.7,9.7) -- (.3,9.7) -- cycle;
\draw[thick,black] (.3,7.3) -- (2.7,9.7);

\draw (0,0) grid (3,10);

\end{tikzpicture}
\end{center}

\caption{A ruler rectangle of size $\phi(3,3) \times 3$.}
\label{fig:Rectangles}
\end{figure}
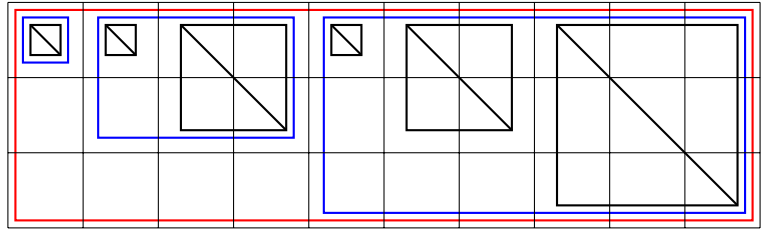

Above the dedicated half-line we put $|p(x)|$ sequences of \emph{data rectangles}, one for each $y \in p(x)$, stacked on top of each other. The data rectangles of $y$ contain patterns from the configuration $f(y)$. The left and right ends of the data rectangles are forced to align with those of the ruler rectangles, and the heights increase by the respective $k(y)$ every step. If $y \in M$, the height of the $n$th rectangle of the sequence is $n$, and it will be filled with the unary pattern of $f(y)$. If $y \notin M$, it is created using this construction, and consists of a finite number of horizontal sequences of rectangles whose total width increases by a constant $k(y)$ every step. For each $n \in \N$, the $n$th data rectangle of the sequence corresponding to $y$ has height $n \cdot k(y)$, and it is forced to contain a pattern of $f(y)$ aligned with the right border of the rectangle as in Figure~\ref{fig:RectAlign}. The linear growth is easily forced by SFT rules. In the construction, each $f(x)$ will extend the alphabet with completely new symbols (apart from the ones used to simulate the $f(y)$), and each region in the construction will have a different unary background symbol to differentiate them from each other.

\begin{figure}[ht]

\begin{center}
\begin{tikzpicture}

\fill[black!30!white] (1.5,.1) rectangle (3,.85);
\fill[black!30!white] (1.5,1.6) rectangle (3,.85);
\fill[black!30!white] (.7,.4) rectangle (1.5,.85);
\fill[black!30!white] (.7,1.3) rectangle (1.5,.85);
\fill[black!30!white] (.5,.7) rectangle (.7,.85);
\fill[black!30!white] (.5,1) rectangle (.7,.85);

\begin{scope}[shift={(4,.25)}]
\fill[black!30!white] (3,-.15) rectangle (5,.85);
\fill[black!30!white] (3,1.85) rectangle (5,.85);
\fill[black!30!white] (1.5,.1) rectangle (3,.85);
\fill[black!30!white] (1.5,1.6) rectangle (3,.85);
\fill[black!30!white] (.7,.4) rectangle (1.5,.85);
\fill[black!30!white] (.7,1.3) rectangle (1.5,.85);
\fill[black!30!white] (.5,.7) rectangle (.7,.85);
\fill[black!30!white] (.5,1) rectangle (.7,.85);
\end{scope}

\draw (1.5,.1) rectangle (3,.85);
\draw (1.5,1.6) rectangle (3,.85);
\draw (.7,.4) rectangle (1.5,.85);
\draw (.7,1.3) rectangle (1.5,.85);
\draw (.5,.7) rectangle (.7,.85);
\draw (.5,1) rectangle (.7,.85);

\begin{scope}[shift={(4,.25)}]
\draw (3,-.15) rectangle (5,.85);
\draw (3,1.85) rectangle (5,.85);
\draw (1.5,.1) rectangle (3,.85);
\draw (1.5,1.6) rectangle (3,.85);
\draw (.7,.4) rectangle (1.5,.85);
\draw (.7,1.3) rectangle (1.5,.85);
\draw (.5,.7) rectangle (.7,.85);
\draw (.5,1) rectangle (.7,.85);
\end{scope}

\draw[thick,dashed] (0,0) -- (10,0);

\draw[thick] (0,.1) -- (3,.1) -- (3,1.6) -- (0,1.6);
\draw[thick] (3,.1) -- (9,.1) -- (9,2.1) -- (3,2.1) -- (3,1.6);
\draw[thick] (9,.1) -- (10,.1);
\draw[thick] (9,2.1) -- (9,2.6) -- (10,2.6);

\draw[thick] (0,-.1) -- (10,-.1);
\draw[thick] (3,-.1) -- (3,-.3);
\draw[thick] (9,-.1) -- (9,-.3);

\end{tikzpicture}
\end{center}

\caption{The alignment of $f(y)$ in the data rectangles. The ruler rectangles lie below the horizontal line.}
\label{fig:RectAlign}
\end{figure}
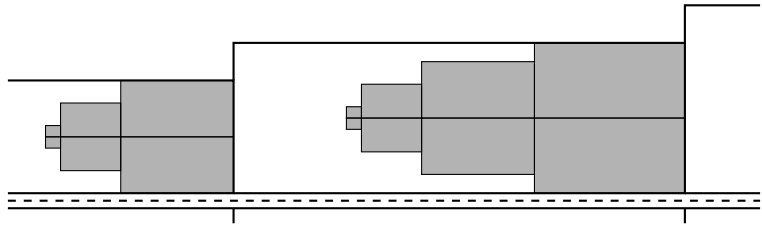

That the construction can be done using only SFT rules is clear, as is the fact that $f$ becomes an order-embedding of $(S,\geq)$ into $(X,\succcurlyeq)$. Furthermore, a simple case analysis shows that $X$ is countable and its subpattern poset is finite. 
\end{proof}

\section{Conclusions}

In this paper, we have presented several constructions related to the computational and topological structure of countable sofic and finite type subshifts. In Theorem~\ref{thm:DerivativesCanCompute}, we presented a single countable SFT whose $k$th derivative is $\PI^0_{2k+1}$-complete for all $k$, the highest possible among the $k$th derivatives of all $\PI^0_1$ subshifts. We also studied the subpattern posets of countable subshifts, our main result being Theorem~\ref{thm:SoficChain}. The theorem is much more interesting in conjunction with Conjecture~\ref{conj:NoDown} than in itself, since if the conjecture is true, the sofic counterexample may be helpful in finding a proof for it.

We have achieved the exact maximal computational strength of the $k$th derivative of a countable SFT, but it would be interesting to see what happens in the first limit ordinal and beyond.

\begin{question}
Let $\lambda$ be any computable ordinal. What is the maximal computational power of the $\lambda$th derivative of a countable SFT? 
\end{question}

In particular, does there exist a countable SFT whose $\omega$th derivative is $\PI^0_k$-hard for all $k$? Can we reach higher levels of the hyperarithmetical hierarchy this way? Since we are not experts in recursion theory, it may be the case that these questions have already been answered in some form, but we are simply not aware of these results.

We also have some interesting open questions regarding the subpattern posets of subshifts. We have shown here that all finite posets can be order-embedded in the subpattern poset of some countable SFT. For which infinite posets does this hold? As a more concrete question, let $\F = \{ A \subset \N \;|\; |A| < \infty \}$, and consider the partial order $<$ defined on $\F$ by $A < B$ iff $B \subset A$.

\begin{question}
Let $\F' \subset \F$ be computable. Can the induced poset $(\F',<)$ be order-embedded in the subpattern poset of some countable SFT or sofic shift?
\end{question}

Example~\ref{ex:YesSide} shows that the poset $(\{ \{n\} \;|\; n \in \N \}, <)$ can be order-embedded in a countable SFT. Also, Theorem~\ref{thm:SoficChain} shows that the poset $(\{ \{1,\ldots,n\} \;|\; n \in \N \}, <)$ can be order-embedded in a countable sofic shift, while Conjecture~\ref{conj:NoDown} claims that this is impossible for an SFT. Note that the above question does not clash with Proposition~\ref{prop:NoUp}, since we have inverted the subset relation.

\bibliographystyle{eptcs}
\bibliography{bib}{}

\end{document}